\documentclass[conference]{IEEEtran}

\usepackage{graphicx}
\usepackage{amssymb}
\usepackage{amsmath}
\usepackage{caption}
\usepackage{subcaption}
\usepackage{algorithm}
\usepackage{amsthm}
\usepackage{comment}
\usepackage{cite}
\usepackage{mathrsfs}
\usepackage{float}
\ifCLASSINFOpdf
\else
\fi
\usepackage{amsmath}

\newtheorem{theorem}{Theorem}
\theoremstyle{definition}
\newtheorem{defn}{Definition} 

\newtheorem{rem}{Remark}
\begin{document}
\title{Storage-Latency Trade-off in Cache-Aided\\ Fog Radio Access Networks}
\providecommand{\keywords}[1]{\textbf{\textit{Index terms---}} #1}

\author{\IEEEauthorblockN{Joan S. Pujol Roig}
\IEEEauthorblockA{
Imperial College  London\\
jp5215@imperial.ac.uk}

\and

\IEEEauthorblockN{Filippo Tosato}
\IEEEauthorblockA{Toshiba Research Europe\\
filippo.tosato@toshiba-trel.com}
\and
\IEEEauthorblockN{Deniz G\"{u}nd\"{u}z}
\IEEEauthorblockA{
Imperial College  London\\
d.gunduz@imperial.ac.uk}}

\maketitle

\begin{abstract} %\makeatletter{\renewcommand*{\@makefnmark}{}
%\footnotetext{This work was partially funded by the European Research Council through Starting Grant BEACON (agreement No. 677854).}\makeatother
%}

A fog radio access network (F-RAN) is studied, in which $K_T$ edge nodes (ENs) connected to a cloud server via orthogonal fronthaul links, serve $K_R$ users through a wireless Gaussian interference channel. Both the ENs and the users have finite-capacity cache memories, which are filled before the user demands are revealed. While a centralized placement phase is used for the ENs, which model static base stations, a decentralized placement is leveraged for the mobile users. An achievable transmission scheme is presented, which employs a combination of interference alignment, zero-forcing and interference cancellation techniques in the delivery phase, and the \textit{normalized delivery time} (NDT), which captures the worst-case latency, is analyzed. 
%The proposed transmission scheme considers the interplay between ENs' caches, users' caches and the transmission rates over the backhaul links, and is studied for both \textit{serial} and \textit{pipelined} transmissions. %Performance comparison with existing literature is also provided.
\end{abstract}

%\keywords{F-RAN, cache-aided networks, decentralized caching, interference management, normalized delivery time.}

\IEEEpeerreviewmaketitle
\section{Introduction}
%With the explosion of multimedia content and social networks, wireless network traffic is experiencing a deep transformation as it is increasingly becoming dominated by video content. To exploit the features of the video traffic and to cope with the exponential growth of future network traffic, \textit{proactive content caching} has emerged as a promising technique. In cache-aided networks, popular contents are placed in memories distributed across the network during low traffic periods (\textit{placement phase}), and the goal is to use these memories to serve the users during peak traffic periods (\textit{delivery phase}) so that caching can replace, or reduce, the need for back-haul communications.

In their pioneering work \cite{maddah2014fundamental}, Maddah-Ali and Niesen showed that proactive caching at user terminals combined with coded delivery over an error-free shared link can significantly reduce the amount of data that needs to be transmitted over the shared link, compared to traditional uncoded caching and unicast delivery of demands. They proposed a novel centralized coded caching scheme, which creates and exploits multicasting opportunities across users, significantly reducing the required delivery rate. Benefits of coded caching extend to the decentralized setting, where users cache bits independently from one another \cite{maddah2015decentralized}, \cite{amiri2016decentralized}.

An architecturally dual setting is considered in \cite{maddah2015cache}, where caching is employed at the transmitter side. In this model multiple cache-aided transmitters deliver content over a wireless channel. In  \cite{naderializadeh2016fundamental}, authors address interference management in a cache-aided network with an arbitrary number of cache-enabled transmitters and users. The proposed delivery scheme makes use of zero-forcing (ZF) techniques as well as interference cancellation (IC) to satisfy users' demands. A constant-factor approximation to the \textit{sum degrees-of-freedom} ($\mathrm{sDoF}$) in a $K_T\times K_R$  cache-aided interference network with caches at both ends is provided in \cite{ursniesen}, by using a combination of interference alignment (IA) and IC techniques. Cache-aided interference networks with caches at both ends are studied in \cite{xu2016fundamental} for centralized cache placement, and in \cite{roig2017interference} for centralized cache placement at the transmitters and decentralized cache placement at the users.

\par
Note that in the interference channel model studied in the aforementioned papers, transmitters must be capable of caching all the library collectively to be able to satisfy all demand combinations. Instead, in the cloud-aided fog radio access network (F-RAN) model studied in \cite{sengupta2016cloud}, edge nodes (ENs) can fetch contents from the cloud through finite-capacity frounthaul links. The normalized delivery time (NDT) is studied in \cite{sengupta2016cloud} by exploiting a centralized placement phase and a delivery phase that leverages ENs' caches as well as the cloud links. In \cite{girgis2017decentralized}, an F-RAN architecture is considered with decentralized cache placement at both the ENs and the users, and an achievable scheme is proposed for two ENs and an arbitrary number of users. The authors in \cite{ding2017network} characterize the achievable NDT for an F-RAN with a shared cloud link and centralized cache placement at both the ENs and users. 
\par

%%%%%%%%%%%%%%%%%%%%%%%%%%
\begin{figure}
\centering
\includegraphics[width=0.5\textwidth]{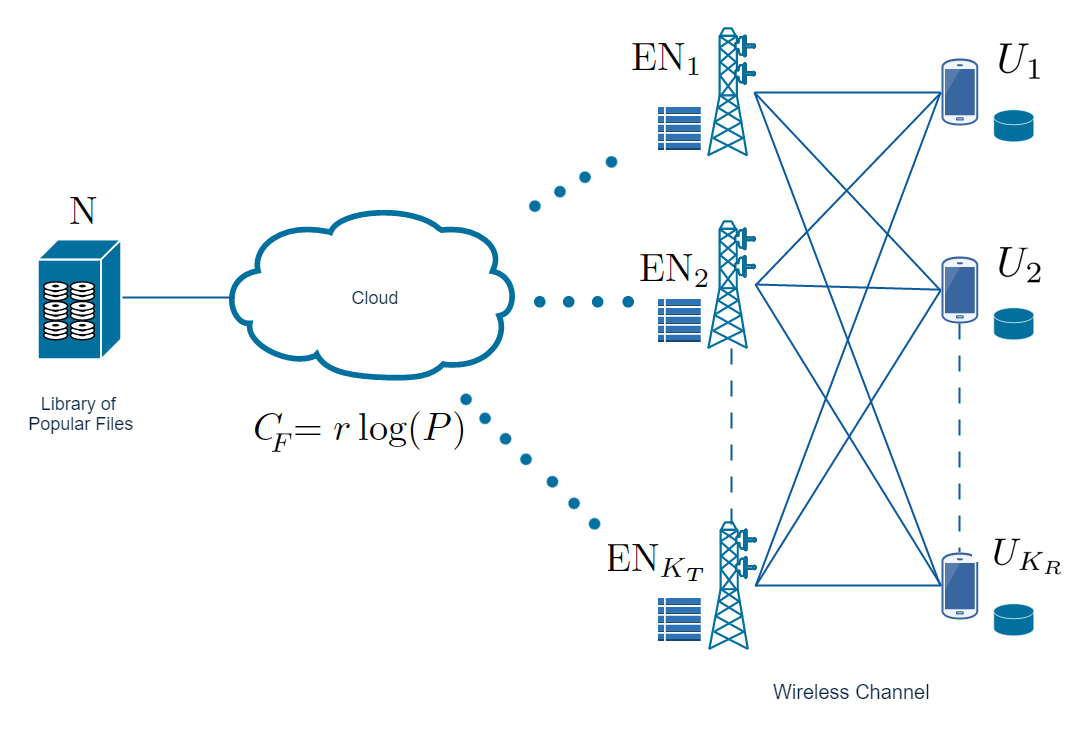}
\caption{The $K_T \times K_R$ cloud- and cache-aided F-RAN architecture with caches at both the ENs and the users.}
\label{figura 1}
\end{figure}
%%%%%%%%%%%%%%%%%%%%%%%%%%

In this work, we consider an F-RAN consisting of single antenna terminals with cache capabilities at both the ENs and the users. Our model considers decentralized placement at the users' caches, while caching at the ENs is centralized. Centralized coordination of the cache contents at the ENs, which model fixed base stations, is a reasonable assumption, while decentralized cache placement is needed for mobile users roaming around. We propose a new decentralized delivery scheme for F-RANs based on the decentralized cache placement ideas presented in \cite{roig2017interference} and the \textit{soft-transfer} delivery scheme of \cite{sengupta2016cloud}. This achievable scheme aims to minimize the NDT taking into account the interplay between the ENs' caches, users' caches and the capacity of the cloud links. The proposed delivery scheme jointly exploits IA, ZF, IC as well as the ENs' fronthaul links, and is studied for both \textit{serial} and \textit{pipelined} transmissions.
\par
In comparison with \cite{sengupta2016cloud}, where authors consider an F-RAN with caches only at the ENs, our model also considers caches at the user side, similarly to \cite{girgis2017decentralized,ding2017network}. However, \cite{girgis2017decentralized} considers decentralized placement for all the network's caches, including those at the ENs, and is limited to two ENs; whereas we propose an achievable scheme for an arbitrary number of ENs and users. Unlike \cite{ding2017network}, we leverage a decentralized placement phase for the users, and study dedicated cloud  links to each of the ENs.  Moreover, in contrast to \cite{sengupta2016cloud}, we do not assume knowledge of the capacity of the cloud links during the placement phase, a more realistic assumption since the future back-haul congestion (hence, the cloud link capacity) is unknown during off-peak traffic periods. Finally, our delivery scheme leverages a combination of IA, ZF and IC, compared to exploiting either IA or ZF or IC (in the presence of user caches) as in the delivery schemes of \cite{sengupta2016cloud, ding2017network, girgis2017decentralized}.

%%%%%%%%%%%%%%%%%%%
\section{System Model}\label{s:System_Model}

We consider an F-RAN architecture with $K_T$ ENs, $\mathrm{EN}_{1}, \dots ,\mathrm{EN}_{K_T}$, and $K_R$ users $ \mathrm{U}_1, \dots ,\mathrm{U}_{K_R}$ (see Figure \ref{figura 1}). A cloud server holds a library of $N\geq K_R$ popular files, $\mathbf{W}\triangleq (W_{1}, W_{2},\dots ,W_{N})$, each of size $F$ bits. Each EN and each user is equipped with a cache memory of size $M_TF$ and $M_RF$ bits, respectively. We refer to the global normalized cache size at the ENs and users as $t_T\triangleq K_TM_T/N$ and $t_R\triangleq K_RM_R/N$, respectively, where $t_T \in [0, K_T]$ and $t_R \in [0, K_R]$. Furthermore, each of the ENs is connected to the cloud server via a dedicated fronthaul link of capacity $C_{F}$ bits per use of the wireless channel. 
\par

In the \textbf{\textit{placement phase}}, all the caches in the network are filled without the knowledge of users' demand or the value of $C_F$. The cache contents of $\mathrm{EN}_i$ and $\mathrm{U}_j$ at the end of the placement phase are denoted, respectively, by a binary sequence $P_i$ of length $M_TF$, $\forall i \in [K_T] \triangleq \{1, \ldots, K_T\}$, and a binary sequence $Q_j$ of length $\lfloor M_RF \rfloor$, $\forall j \in [K_R]$. The cache placement function that maps the library to the EN cache contents in a \textbf{centralized} manner is known by all the ENs, while each EN knows only the contents of its own cache. On the other hand, users leverage a \textbf{decentralized} placement phase, and each user caches an equal number of bits randomly from each file in the library. The number of users that will take part in the delivery process as well as their cache sizes is unknown during this phase. \par

The users reveal their requests at the beginning of the \textbf{\textit{delivery phase}}. Let $W_{d_j}$ denote the file requested by $\mathrm{U}_j ,\ \forall j \in [K_R]$, and $\textbf{d}\triangleq [d_1,\dots, d_{K_R}] \in [N]^{K_R}$ denote the demand vector. The delivery phase takes place over an independent and identically distributed additive white Gaussian noise interference channel. The signal received at $\mathrm{U}_j$ at time $t$ is:
\begin{equation}
Y_j(t)=\sum_{i=1}^{K_T}h_{ji}X_i(t)+Z_j(t),
\end{equation}
where $X_i(t)\in \mathbb{C}$ represents the signal transmitted by $\mathrm{EN}_{i}$, $h_{ji}\in \mathbb{C}$ represents the channel coefficient between user $j$ and $\mathrm{EN}_i$, and $Z_j(t)$ is the additive Gaussian noise term at $\mathrm{U}_j$. We assume that the channel coefficients $\mathbf{H} \triangleq \{h_{i,j}\}_{i\in[K_R], j\in[K_T]}$, and the demand vector $\textbf{d}$ are known by all the ENs and users. 
\par
The cloud server maps the demand vector $\textbf{d}$, the library $\mathbf{W}$ and the channel matrix $\mathbf{H}$ to message $\mathbf{U}_i$ of length $L_F$,  $\mathbf{U}_i \triangleq [U_i(1), \ldots, U_i(L_F)]$, for $i \in [K_T]$, which is sent to $\mathrm{EN}_i$ through the fronthaul link. $L_F$ is normalized to the symbol transmission duration over the downlink wireless channel; and therefore, the message $\mathbf{U}_i$ to $\mathrm{EN}_i$ is limited to $L_FC_F$ bits. $\mathrm{EN}_i$, $\forall i\in [K_T]$, maps $\mathbf{d}$, $\mathbf{U}_i$, $\mathbf{H}$, and its own cache contents $P_i$ to a channel input vector of length $L_E$, $\mathbf{X}_i = [X_i(1), \ldots, X_i(L_E)]$. We impose an average power constraint $P$ on each transmitted codeword, i.e., $\frac{1}{L_E}  \| \mathbf{X}_i\| ^2 \leq P$. 

User $\mathrm{U}_j$, $\forall j\in [K_R]$, decodes its desired file $W_{d_j}$ using $\mathbf{d}, \mathbf{H}$, its own cache content $Q_j$, and the corresponding channel output $\mathbf{Y}_j= [Y_j(1), \ldots, Y_j(L_E)]$. Let $\hat{W}_{j}$ denote its estimate of $W_{d_j}$.The error probability is defined as:
\begin{align}
    P_e = \max_{\mathbf{d}\in [N]^{K_R}} \max_{j \in [K_R]} \mathrm{Pr}\left( \hat{W}_{j} \neq W_{d_j}\right).
\end{align} 
We now introduce the performance measure, NDT, which accounts for the worst-case latency in the delivery phase \cite{zhang2017fundamental}, \cite{sengupta2016cloud}.

\begin{defn}\label{def1}
Delivery time per bit $\Delta	(t_T,t_R,C_F,P)$ is \textit{achievable}, if there exists a sequence of codes, indexed by file size $F$, such that $P_e \rightarrow 0$ as $F \rightarrow \infty$, and
\begin{equation}
\Delta	(t_T,t_R,C_F,P)=\liminf _{F \rightarrow \infty}\frac{T(L_F,L_E)}{F},
\end{equation}
where $T(L_E,L_F)$ accounts for the end-to-end latency, and depends on the transmission approach considered (see Definitions \ref{definition2} and \ref{definition3} below).
\end{defn}

\begin{defn} \cite{sengupta2016cloud} 
For a given family of codes achieving a delivery time per bit of $\Delta	(t_T,t_R,C_F,P)$, and a fronthaul link capacity that scales as $C_F=r\log P$, the \textit{normalized delivery time} (NDT) of the family of codes in the high SNR regime is defined as:
\begin{equation}
\delta (t_T,t_R,r) \triangleq \lim_{P \rightarrow \infty}\frac{\Delta	(t_T,t_R,r\log P,P)}{1/\log P}.
\end{equation}
\end{defn}

%The minimum NDT for given $(M_T,M_R,r)$ is defined as:
%\begin{align*}
%\delta ^*(M_T,M_R,r) \triangleq \inf \{ \delta(M_T,M_R,r): 
%\delta(M_T,M_R,r)\ is \ \mbox{achievable}\}.
%\end{align*}

Our goal in this paper is to characterize the minimum achievable NDT for a given network. We will refer to the fronthaul-NDT as $\delta _F$ and the edge-NDT as $\delta _E$, which are determined by $L_F$ and $L_E$, respectively. Following \cite{sengupta2016cloud}, we study two types of transmission approaches \textit{serial} and \textit{pipelined}, as explained below.

\begin{defn}\label{definition2} 
In \textit{serial transmission}, the fronthaul and edge transmissions occur successively; that is, first the cloud server transmits all the $U_i$ messages to the ENs, after which the transmission of the $X_i$ messages over the wireless channel starts, so that $T(L_F,L_E)=L_F+L_E$. Hence, the NDT is given by $\delta _S= \delta _F + \delta _E$.
\end{defn}

\begin{defn}\label{definition3} 
In \textit{pipelined  transmission},  the ENs can simultaneously receive information from the cloud  server through the fronthaul links, and transmit information to the users through the wireless channel. Thus, $\mathrm{EN}_i$ can start the transmission of $X_i$  before the reception of $U_i$ is completed. Using the strategy defined in \cite{sengupta2016cloud} for this model of transmission, we have $T(L_F,L_E)=\max \{ L_F,L_E\}$, and the NDT  is given by $\delta _P= \max \{ \delta _F, \delta _E \}$.
\end{defn}

%%%%%%%%%%%%%%%%%%%%%%%%
%%%%%%%%%%%%%%%%%%%%%%%%
\section{Proposed caching and transmission scheme}\label{s:Centralized}
In this section, an achievable scheme for a $K_T\times K_R$ cache-aided F-RAN with centralized cache placement at the ENs and decentralized cache placement at the users is proposed.
\subsection{Placement Phase}\label{s:Decentralized}
The users leverage a decentralized placement phase, which allows us to exploit coded delivery without relying on centralized planning of the cache contents. To this end, each user fills its cache with randomly chosen $M_RF/N$ bits of each file, so that the cache capacity constraint is met. On the other hand, the ENs, which correspond to stationary base stations, leverage the following centralized placement scheme (see Figure \ref{placement}): when $t_T<1$ each $EN_i,i\in [K_T]$, stores $M_TF/N$ non-overlapping bits from each file in the library, and the remainders of the files are accessible only from the cloud server through the fronthaul links. On the other hand, when $t_T\geq1$, each file of the library is split into two parts, such that, one of the parts is stored by all the ENs while the other part is stored collectively across the ENs (each EN caches a distinct part). As a result, each EN stores $(1-M_T/N)F/(K_T-1)$ non-overlapping bits of  each file of the library plus the same $(t_T-1)F/(K_T-1)$ bits of each file, fulfilling the memory size constraint. Unlike \cite{sengupta2016cloud}, the fronthaul link capacity is assumed unknown during the placement phase; therefore, the placement cannot be optimized based on $C_F$. This also means that the delivery when $t_T <1$ is not feasible if $C_F=0$.

After the placement phase, if $M_R>0$, each file in the library is further divided into $2^{K_R}$  \textit{subfiles}. We denote the subfile of file $i \in [N]$ stored at $EN_k$, $\forall k \in \mathcal{S}_T$, and at users $K_j$, $\forall j \in \mathcal{S}_R$, by $W_{i,\mathcal{S}_T,\mathcal{S}_R}$, where $\mathcal{S}_T\subset [K_T]$, with size $|\mathcal{S}_T | \in \left\lbrace 1, K_T\right\rbrace$, and  $\mathcal{S}_R\subset [K_R]$ of size $|\mathcal{S}_R |\in [ K_R ]$. Consider, for example, $K_R=3,\ M_R=1,\ K_T=3,\ |\mathcal{S}_T |=1$ and $N=3$. According to the placement phase explained above, file $W_1$ is divided into 24 subfiles as follows:
\begin{eqnarray*}
W_{1,1,\emptyset},\ W_{1,1,1},\ W_{1,1,2},\ W_{1,1,3},\ W_{1,1,12},\ W_{1,1,13},\\ W_{1,1,23},\ W_{1,1,123},\ W_{1,2,\emptyset},\ W_{1,2,1}\ W_{1,2,2},\ W_{1,2,3},\\ W_{1,2,12},\ W_{1,2,13},\  W_{1,2,23},\ W_{1,2,123},\ W_{1,3,\emptyset},\ W_{1,3,1},\\ W_{1,3,2},\ W_{1,3,3},\ W_{1,3,12},\ W_{1,3,13},\ W_{1,3,23}, W_{1,3,123}.
\end{eqnarray*}
% subfile $W_{1,12,\emptyset}$ represents the subfile of file $W_1$ stored at $\mathrm{EN}_1$ and $\mathrm{EN}_2$, and at none of the users, while 
In the previous notation, subfile $W_{1,1,13}$ denotes the subfile of file $W_1$ stored at $\mathrm{EN}_1$, and users $\mathrm{U}_{1}$ and $\mathrm{U}_{3}$. Same partition applies to files $W_2$ and $W_3$. 
\begin{rem}\label{remark1}By the law of large numbers, the size of the subfile that is stored by $j$ out of $K_R$ users, each of them  caching $M_RF/N$ bits from that file, can be approximated by
\begin{align}
    F'(j)\approx \left(\frac{M_R}{N}\right)^j\left( 1-\frac{M_R}{N}\right)^{K_R-j}F\ \mbox{bits}.
\end{align}
\end{rem}
 The total size of the subfiles of a file that need to be transmitted to a user requesting that file in the \textit{delivery phase}, that is, the subfiles which have not been stored in the cache of the requesting user, is given by
\begin{align}
\label{e:size}
\begin{split}
   F'_{r}\triangleq \sum _{j=0}^{K_R-1}\binom{K_R -1}{j}F'(j)\ bits.
\end{split}
\end{align}

\begin{figure}
\centering
\includegraphics[width=0.4\textwidth]{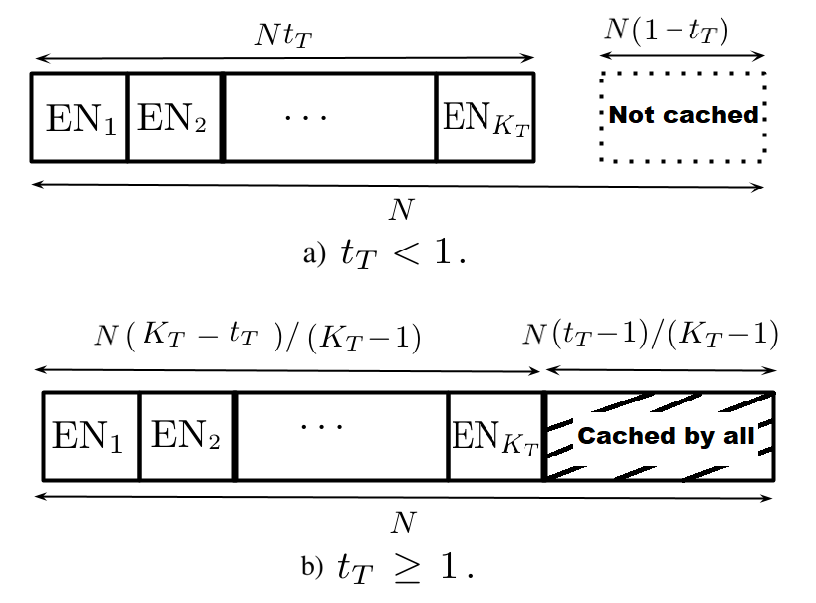}
\caption{EN placement phase.}
\label{placement}
\end{figure}

We reemphasize that the proposed placement phase is independent of the user cache capacities, or the fronthaul links capacities. The proposed delivery scheme exploits a combination of IA and IC, or ZF and IC (similarly to the the decentralized transmission approach in \cite{roig2017interference}).\par  
We highlight the following observations that come from the availability of caches at the users. When implementing IA, we can exploit the cache contents of the users to reduce the number of interfering dimensions at each user. Furthermore, each of the subfiles of a file will achieve a different NDT, i.e., the subfiles that are cached at a single user will achieve a higher NDT than those cached at $K_R-1$ users. The following expression provides the NDT  achieved for the delivery of subfiles stored in $j$ out of $K_R$ users, using a combination of IA and IC:
\begin{equation} \label{eqmax}
\delta _{IA}(j)=\frac{\binom{K_R -1}{j}\ K_R}{\max \left\lbrace \frac{K_TK_R}{K_T+K_R-j},j+1\right\rbrace}F'(j).
\end{equation}
In the numerator we have the total size of the subfiles that will be transmitted, while the denominator is the achievable \textit{sum degrees-of-freedom}. The first argument in the $\max$ in (\ref{eqmax}) corresponds to the well-known expression of the DoF achievable by IA in an X-channel.  If the subfiles are carefully grouped (as in \cite{roig2017interference}) for transmission,  the number of interfering dimensions can be reduced by $j$ thanks to the users' cache contents. The second argument corresponds to the joint transmission of subfiles. Consider, for example, the subfiles $W_{1,1,2}$ and $W_{2,2,1}$, requested by $\mathrm{U}_1$ and $\mathrm{U}_2$, respectively. These subfiles can be transmitted simultaneously, as $\mathrm{U}_1$ can cancel $W_{2,2,1}$ (available in its cache) and $\mathrm{U}_2$  can cancel $W_{1,1,2}$. \par
With ZF, the users' cache contents play a similar role, the number of users at which interference can be nullified is increased due to the side information available at the users. As a result, the following expression  provides an upper-bound on the NDT  for the transmission of the subfiles cached by $K_T$ transmitters and $j$ out of $K_R$ users, leveraging a combination of ZF and IC:
\begin{equation}\label{eqnzf}
\delta _{ZF}(j)=\frac{\binom{K_R -1}{j}\ K_R}{\min \left\lbrace K_T+j,K_R\right\rbrace}F'(j).
\end{equation}
Again, the numerator in (\ref{eqnzf}) corresponds to the total size of the subfiles that must be transmitted, while the denominator corresponds to the DoF. If the files to be transmitted are carefully selected, the ENs, which share the same information, can reduce the number of interfering signals at the users by $j$. Consider, for example, subfiles $W_{1,1,2}$, $W_{2,2,3}$ and $W_{3,3,1}$, requested by $\mathrm{U}_1$, $\mathrm{U}_2$ and $\mathrm{U}_3$, respectively. These subfiles can be transmitted simultaneously, and by ZF we can cancel $W_{2,2,3}$ at $\mathrm{U}_1$, $W_{1,1,2}$ at $\mathrm{U}_3$ and $W_{3,3,1}$ at $\mathrm{U}_2$. The interfering subfiles are available at the interfered user caches, so these interferences can be canceled. As a result, the desired subfiles are received interference-free with a DoF of $K_T+j$.\par
In the proposed placement phase (Section \ref{s:Decentralized}), if $t_T\geq 1$, the database is divided into two parts. The first part is divided into subfiles which are collectively cached across all the ENs, and the second part is  cached by all the ENs. The transmission from  ENs to users is carried out as a combination of IA-IC and ZF-IC, for these two parts, respectively, which achieves the following NDT:
\begin{equation}
\delta _{ZF-IA}=\sum _{j=0}^{K_R-1}\left( \frac{K_T-t_T}{K_T-1}\delta _{IA}(j)+\frac{t_T-1}{K_T-1}\delta _{ZF}(j) \right).
\end{equation}
\subsection{Delivery Phase}\label{s:delivery}
Next, we present the proposed delivery scheme for \textit{serial transmission}. All user demands must be satisfied by the end of the delivery phase.  In the rest of the paper, we assume that  each user requests a different file from the library, corresponding to the worst-case demand combination.
\par

\textbf{Edge-Only Delivery:} When fronthaul links are not available, i.e., $r=0$, all demands must be satisfied from the EN and user caches, requiring $t_T \geq 1$.  We remark that, during the placement phase, we do not know the fronthaul link capacities; and moreover, due to  decentralized cache placement we cannot guarantee any of the bits to be available at user caches; hence the requirement $t_T\geq 1$. By exploiting edge-only delivery; we can achieve an NDT of
\begin{equation}
\delta ^{e}=\delta ^e_{E}=\delta ^e_{ZF-IA},
\end{equation}
which is obtained using the combination of IA-IC and ZF-IC transmission techniques.

\textbf{Cloud-Only Delivery:} Cloud-only delivery is used  when there are no caches at the ENs, i.e.,  $t_T=0$. This requires a non-zero fronthaul link capacity, i.e., $r>0$.  For this particular network configuration, the following NDT is achievable:
\begin{equation}
\delta ^c=\delta ^c_{E}+\delta ^c_{F},
\end{equation}
where
\begin{equation*}
\begin{split}
\delta ^c_{E}&= \sum _{j=0}^{K_R-1}\frac{K_R\ \binom{K_R-1}{j}}{\min (K_R,K_T+j)}F'(j)\\
\delta ^c_{F}&=\frac{K_R }{K_T r}F_r'
\end{split}.
\end{equation*}
Where $F_r'$ is as defined in (\ref{e:size}). This NDT is achieved by using the \textit{soft-transfer mode} proposed in \cite{sengupta2016cloud} to transmit the remaining $F_r'$ bits of each of the $K_R$ requested files,  where the cloud server implements ZF-beamforming and the resulting encoded signals are quantized and transmitted to the ENs. 

\textbf{Joint Edge and Cloud-Aided Delivery:}  In general, both the fronthaul links and the EN caches should be used to deliver the requested files, when the ENs cannot store the whole database collectively ($0<t_T< 1$) . With the placement phase of Section  \ref{s:Decentralized},  part of the requested files are available in each of the ENs, while the rest of them will be sent through the fronthaul links. The subfiles that are available  at the EN caches are transmitted using the IA and IC techniques (Section \ref{s:Decentralized}), and the rest through the \textit{soft-transfer} scheme.
Therefore, the achievable NDT is given by:
\begin{align}
\delta ^{h}=\sum _{j=0}^{K_R-1}t_T\delta _{IA}(j)+(1-t_T)\delta ^c,
\end{align}
where
\begin{equation}
\begin{cases}
\delta ^{h}_E&=\sum _{j=0}^{K_R-1}t_T\delta _{IA}(j) +(1-t_T)\delta ^c_{E} \\
\delta ^{h}_F&=(1-t_T)\delta ^c_{F}
\end{cases}.
\end{equation}

\section{Main Results}
The following theorems provide an upper-bound on the achievable NDT  for \textit{serial} and \textit{pipelined} transmissions. 
\begin{theorem}
For a $K_T \times K_R$ F-RAN with centralized placement at the ENs and decentralized placement at the users, and a fronthaul capacity of $r\geq 0$, the following NDT is achievable with serial transmission:
\end{theorem}
\begin{equation}
\delta _S=
\begin{cases}
\min \{ \delta ^{h}, \delta^{c}\}\text{ if }t_T\leq 1\\
\min \{ \delta ^{e}, \delta^{c}\}\text{ if }t_T\geq 1\\
\end{cases}.
\end{equation}

\begin{proof}
In serial transmission, the total NDT is the sum of the fronthaul  ($\delta_F$) and edge ($\delta_E$) delays, which corresponds to the minimum of the NDTs of the \textit{cloud-only delivery} or \textit{edge only delivery} schemes when $t_T<1$; and the minimum of the NDTs of the \textit{cloud-only delivery} or \textit{joint edge and cloud-aided delivery} when $t_T\geq 1$.  Once the fronthaul link capacity is revealed, the best transmission scheme is chosen based on the fronthaul rate and the EN cache size. If fronthaul rates are low, e.g., high network congestion,  \textit{edge-only} delivery will be leveraged if $t_T<1$, or \textit{joint edge and cloud-aided} delivery if $t_T\geq 1$. On the other hand, if the fronthaul capacity is high, \textit{cloud-only} approach outperforms the two other schemes.
\end{proof}

\begin{theorem}
For a $K_T \times K_R$ F-RAN with centralized placement at the ENs and decentralized placement at the users, and a fronthaul link capacity of $r \geq 0$, the following NDT with pipelined transmission is achievable:
\begin{equation}
\delta_P=
\begin{cases}
\min \{\max \{ \delta ^{h}_F, \delta^{h}_E\},\max \{\delta^{c}_F, \delta^{c}_E\}\}&\text{ if }t_T\leq 1\\
\min \{\max \{ \delta ^{c}_E, \delta^{c}_E\},\delta^{e}_E\}&\text{ if }t_T\geq 1
\end{cases}. \nonumber
\end{equation}
\end{theorem}

\begin{proof}
From the results in \cite{sengupta2016cloud} for this type of transmission, we only need to prove the achievability of the fronthaul and edge delays, which follow from Theorem 1.
\end{proof}

\section{Numerical results}
In this section we present the comparison of the achievable NDT of the proposed caching and delivery scheme with those in \cite{sengupta2016cloud}, referred to as STS, in \cite{ding2017network}, referred to as DYL, and in \cite{girgis2017decentralized}, referred to as GENE, for cloud and cache aided F-RAN.

%----------------
\begin{figure}[]
\centering
\includegraphics[width=0.5\textwidth]{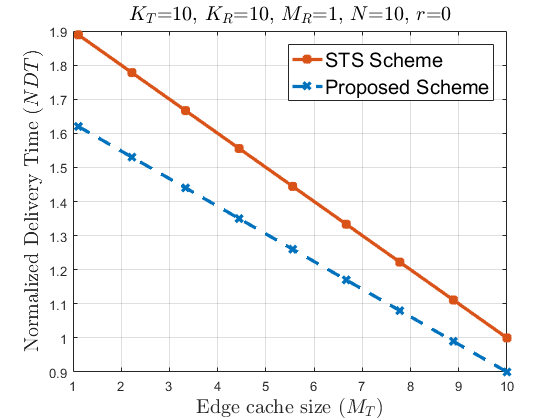}
\caption{NDT vs. $M_T$ for edge-only delivery.}
\label{no-cloud}
\end{figure}
%----------------

We first consider edge-only delivery, i.e., $r=0$, by assuming $M_TK_T\geq N$, or equivalently, $t_T\geq 1$. In Figure \ref{no-cloud} we compare the NDT of the proposed scheme for $M_R=1$ with STS, which does not take advantage of the user caches. The transmission scheme in \cite{girgis2017decentralized} and \cite{ding2017network} are omitted as they require the fronthaul links. The figure illustrates the gains from user caches in terms of the NDT in an F-RAN. We observe that as the EN cache size increases, the performance improvement of the proposed scheme  shrinks. This is because, as $M_T$ increases the delivery scheme exploits ZF, and the benefit of user caches for IC diminishes, and they only account for local caching gain. However, for limited $M_T$ we observe that user caches provide gains beyond local caching gains thanks to combining the IA and ZF techniques with IC. 

%----------------
\begin{figure}[]
\centering
\includegraphics[width=0.5\textwidth]{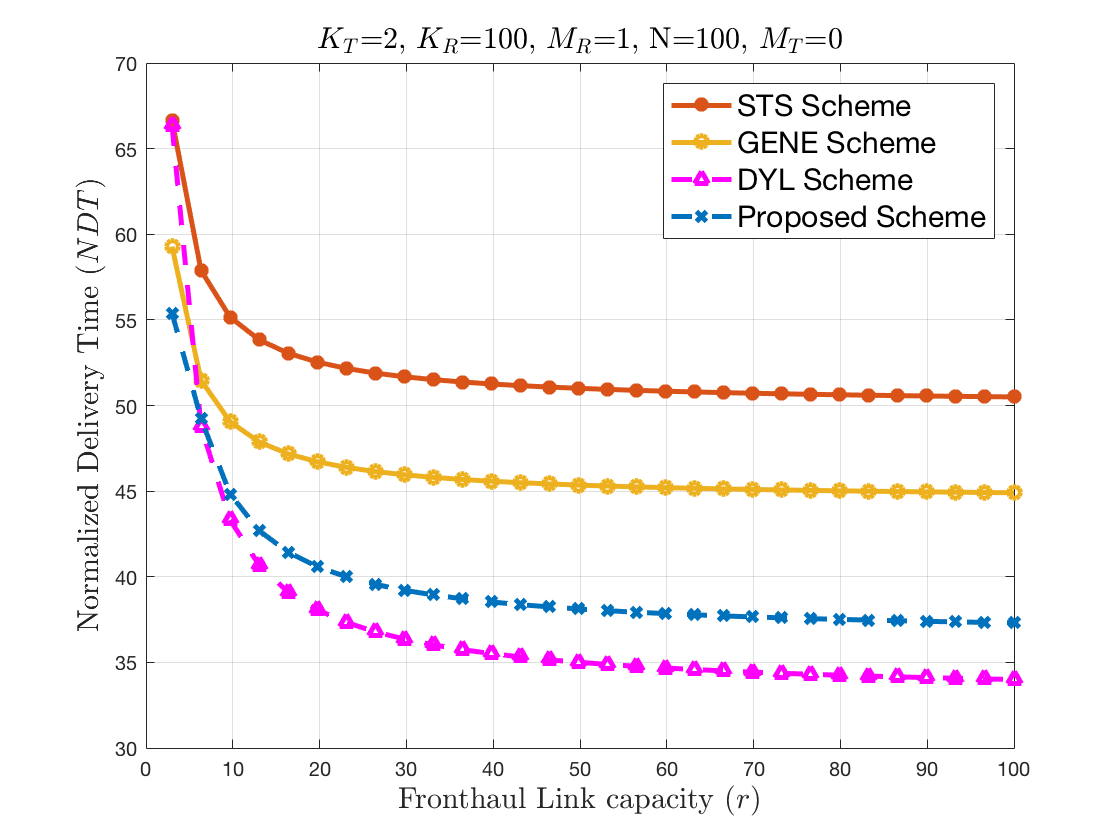}
\caption{NDT vs. founthaul link capacity ($r$).}
\label{no-cahce}
\end{figure}
%----------------

In Figure \ref{no-cahce} we consider cloud-only delivery, i.e., $M_T=0$, with serial transmission. Here, we plot the NDT performance with respect to the fronthaul link capacity $r$. We consider $K_T=2$ to be able to compare the result with that of the GENE scheme. As expected, the NDT decays with $r$, and saturates to a fixed value, which essentially characterizes the edge delay. It must be noted that the STS scheme  of \cite{sengupta2016cloud} does not exploit the user caches, while the GENE scheme assumes decentralized cache placement at the ENs; and hence, their relatively poor performance. The GENE scheme performs poorly compared to the proposed scheme even for high fronthaul link capacities, this is because GENE employs soft-transfer only for the parts of the files that are not cached anywhere in the network, whereas the proposed scheme employs a soft-transfer scheme that enables ZF at the ENs that also benefits from the receiver caches. The DYL scheme instead, exploits centralized cache placement at both the ENs and the users, and as a result it achieves a lower NDT than the other schemes for large enough $r$. The poor performance of the DYL scheme for low $r$ values is due to the shared fronthaul link assumption. 

%The tradeoff lays in the fact that \cite{ding2017network} knows in advance the number and identity of the users that will take part in the delivery phase, and during the placement phase their caches are filled accordingly. 

%----------------
\begin{figure}
\centering
\includegraphics[width=0.5\textwidth]{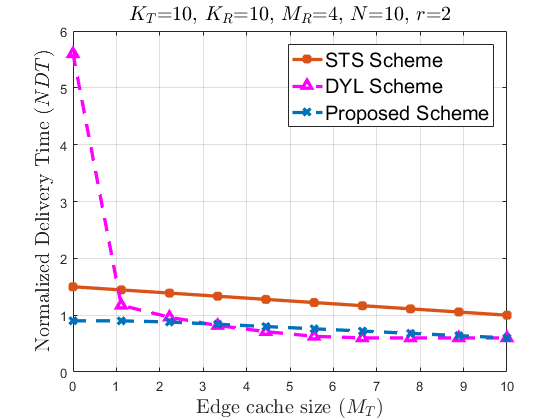}
\caption{NDT vs. $M_T$ for joint cloud and edge delivery.}
\label{general}
\end{figure}
%----------------

Joint edge  and cloud-aided delivery is considered in Figure \ref{general}.  We observe that the performance of the proposed scheme is significantly better than that of the STS scheme, thanks to the user caches, and to the exploitation of  IA, ZF and IC schemes jointly. The performance of DYL is poor at the beginning (due to the low fronthaul capacity), but thanks to the centralized placement of users' caches, it improves with $M_T$. We emphasize that, even though our scheme exploits decentralized placement at user caches, the performance gap with DYL is small. This is thanks to the exploitation of  IA, ZF and IC schemes jointly.

%The initial flat behaviour of the proposed scheme is because at that point the EN cache size is small, so \textit{cloud-only} delivery is used, whose performance does not improve with $M_T$. 

%----------------
\begin{figure}
\centering
\includegraphics[width=0.5\textwidth]{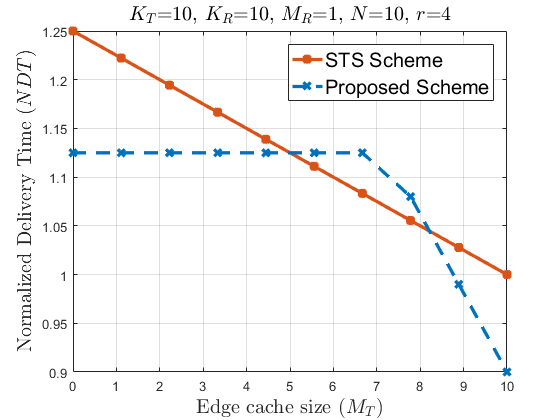}
\caption{NDT vs. $M_T$ for joint cloud and edge delivery.}
\label{general2}
\end{figure}
%----------------

We reemphasize that our scheme does not assume the knowledge of the fronthaul link capacities. This is motivated from the practical consideration that the placement and delivery phases are typically carried out over different time frames, and an accurate prediction of the fronthaul link capacities during the placement phase is too strong an assumption. The consequence of this limitation can be observed in Figure \ref{general2}. Due to the high fronthaul capacity, the STS scheme achieves a lower NDT as $M_T$ increases. The proposed scheme, on the other hand, does not start exploiting the EN caches until $M_T=7$, and employs the soft-transfer scheme before that point, whose performance does not depend on $M_T$ in this case since we have $K_T = K_R$. This is because the proposed scheme uses a placement phase that does not depend on the fronthaul capacity, whereas the STS scheme optimizes the cache placement according to the fronthaul rate.
%----------------
\begin{figure}
\centering
\includegraphics[width=0.5\textwidth]{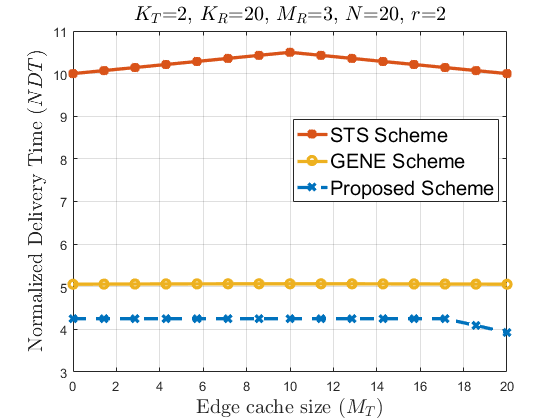}
\caption{NDT vs. MT for joint edge and cloud-aided delivery with pipelined transmission.}
\label{pipe1}
\end{figure}
%----------------
\par
Finally, we compare the NDT under pipelined transmission in Figure \ref{pipe1}. The DYL scheme is omitted as it does not consider pipelined transmission. We observe similar gains as in the serial transmission model, thanks to the user caches and centralized placement at the ENs. 

\section{Conclusions}
We have studied an F-RAN architecture with an arbitrary number of ENs and users, in which both the ENs and the users have cache capabilities. The proposed caching and delivery scheme combines IA, ZF, and IC techniques together with the soft-transfer fronthauling scheme of \cite{sengupta2016cloud}, and a comparison of the achievable NDTs with the existing literature is provided. The proposed scheme takes into account the interplay between the EN caches, user caches, and the fronthaul link capacities, and it is shown to reduce the end-to-end delay significantly for a wide range of system parameters.



%\bibliographystyle{IEEEtran}
%\bibliography{IEEEabrv,../bib/paper} \textit{•}

\bibliographystyle{IEEEtran}
\begin{center}
\bibliography{mybib}

\begin{thebibliography}{10}
\providecommand{\url}[1]{#1}
\csname url@samestyle\endcsname
\providecommand{\newblock}{\relax}
\providecommand{\bibinfo}[2]{#2}
\providecommand{\BIBentrySTDinterwordspacing}{\spaceskip=0pt\relax}
\providecommand{\BIBentryALTinterwordstretchfactor}{4}
\providecommand{\BIBentryALTinterwordspacing}{\spaceskip=\fontdimen2\font plus
\BIBentryALTinterwordstretchfactor\fontdimen3\font minus
  \fontdimen4\font\relax}
\providecommand{\BIBforeignlanguage}[2]{{%
\expandafter\ifx\csname l@#1\endcsname\relax
\typeout{** WARNING: IEEEtran.bst: No hyphenation pattern has been}%
\typeout{** loaded for the language `#1'. Using the pattern for}%
\typeout{** the default language instead.}%
\else
\language=\csname l@#1\endcsname
\fi
#2}}
\providecommand{\BIBdecl}{\relax}
\BIBdecl

\bibitem{maddah2014fundamental}
M.~A. Maddah-Ali and U.~Niesen, ``Fundamental limits of caching,'' \emph{IEEE
  Transactions on Information Theory}, vol.~60, no.~5, pp. 2856--2867, 2014.

\bibitem{maddah2015decentralized}
------, ``Decentralized coded caching attains order-optimal memory-rate
  tradeoff,'' \emph{IEEE/ACM Trans. Networking}, vol.~23, no.~4.

\bibitem{amiri2016decentralized}
M.~M. Amiri, Q.~Yang, and D.~G{\"u}nd{\"u}z, ``Decentralized caching and coded
  delivery with distinct cache capacities,'' \emph{IEEE Trans. on Comm.}, 2017.

\bibitem{maddah2015cache}
M.~A. Maddah-Ali and U.~Niesen, ``Cache-aided interference channels,'' in
  \emph{IEEE Int. Symp. Inf. Theory (ISIT)}, 2015, pp. 809--813.

\bibitem{naderializadeh2016fundamental}
N.~Naderializadeh, M.~A. Maddah-Ali, and A.~S. Avestimehr, ``Fundamental limits
  of cache-aided interference management,'' \emph{IEEE Trans. on Inf. Theory},
  vol.~63, no.~5, pp. 3092--3107, 2017.

\bibitem{ursniesen}
J.~Hachem, U.~Niesen, and S.~Diggavi, ``Degrees of freedom of cache-aided
  wireless interference networks,'' \emph{arXiv:1606.03175}, 2016.

\bibitem{xu2016fundamental}
F.~Xu, M.~Tao, and K.~Liu, ``Fundamental tradeoff between storage and latency
  in cache-aided wireless interference networks,'' \emph{IEEE Trans. on Inf.
  Theory}, 2017.

\bibitem{roig2017interference}
J.~Pujol, D.~G{\"u}nd{\"u}z, and F.~Tosato, ``Interference networks with caches
  at both ends,'' \emph{IEEE Int Conf. on Communications (ICC)}, 2017.

\bibitem{sengupta2016cloud}
A.~Sengupta, R.~Tandon, and O.~Simeone, ``Cloud and cache-aided wireless
  networks: Fundamental latency trade-offs,'' \emph{arXiv:1605.01690}, 2016.

\bibitem{girgis2017decentralized}
A.~Girgis, O.~Ercetin, M.~Nafie, and T.~ElBatt, ``Decentralized coded caching
  in wireless networks: Trade-off between storage and latency,''
  \emph{arXiv:1701.06673}, 2017.

\bibitem{ding2017network}
T.~Ding, X.~Yuan, and S.~C. Liew, ``Network-coded fronthaul transmission for
  cache-aided {C-RAN},'' in \emph{IEEE Int. Symp. Inf. Theory (ISIT)}, 2017,
  pp. 1182--1186.

\bibitem{zhang2017fundamental}
J.~Zhang and P.~Elia, ``Fundamental limits of cache-aided wireless {BC}:
  Interplay of coded-caching and {CSIT} feedback,'' \emph{IEEE Trans. on Inf.
  Theory}, vol.~63, no.~5, pp. 3142--3160, 2017.

\end{thebibliography}
\end{center}
\end{document}